\documentclass{llncs}
\let\origvec\vec

\let\vec\origvec
\usepackage{amsmath}
\usepackage{amssymb}
\usepackage{framed}
\usepackage{graphicx}
\usepackage{etoolbox}
\usepackage[
	n,
	operators,
	advantage, 
	sets,
	adversary,
	landau,
	probability,
	notions,
	logic,
	ff,
	mm,
        primitives,
        events,
        complexity,
        asymptotics, 
        keys]{cryptocode}

\usepackage[matrix,frame,arrow]{xy}
\usepackage{amsmath}
\spnewtheorem*{protocol}{Protocol}{\bfseries}{\itshape}

\let\emptyset\varnothing

\newcommand{\prot}[2]{\ensuremath{\textbf{Protocol}\,\, #1:}\,\,#2\\}
\newcommand{\simu}[2]{\ensuremath{\textbf{Simulator}\,\, #1:}\,\,#2\\}

\renewcommand{\qed}{\hfill$\blacksquare$}

\title{Practical Relativistic Zero-Knowledge for $\mathbf{NP}$}
\author{
Claude Cr\'epeau \inst{1}\! \thanks{Supported in part by FRQNT (INTRIQ) and NSERC.}
\and 
Arnaud %Yoh
Massenet %-Oshima
\inst{2}\! \thanks{This research was performed while a student at McGill University under C.C.'s supervision.}
\and \\
Louis Salvail\inst{3}\! ${}^{\star}$
\and
Lucas %Shigeru 
Stinchcombe\inst{4}\! ${}^{\star\star}$
\and Nan Yang\inst{5}\! \thanks{Supported in part by Professors Jeremy~Clark, and Claude~Cr\'epeau.}
}

\institute{
McGill University, Montr\'eal, Qu\'ebec, Canada.
{crepeau@cs.mcgill.ca}
\and
University of Oxford, Oxford, Oxfordshire, UK.
{arnaud.massenet@mail.mcgill.ca}
\and
Universit\'e de Montr\'eal,
Montr\'eal, Qu\'ebec, Canada.
{salvail@iro.umontreal.ca}
\and
Bloomberg L.P, Tokyo, Japan.
{lucas.stinchcombe@mail.mcgill.ca}
\and
Concordia University, Montreal, Quebec, Canada.
{na\_yan@encs.concordia.ca}
}

\begin{document}

\maketitle

\begin{abstract}
In this work we consider the following problem: in a Multi-Prover environment, how close can we get to prove the validity of an $\mathbf{NP}$ statement in Zero-Knowledge ?
We exhibit a set of two novel Zero-Knowledge protocols for the 3-COLorability problem that use two ({\em local}) provers or  three ({\em entangled}) provers and only require them to reply
two trits each. This greatly improves the ability to prove Zero-Knowledge statements
on very short distances with very minimal equipment.
\end{abstract}

%\input{3COL-2020_s01_intro}

%!TEX root = 3COL-2020.tex
\section{Introduction}

The idea of using distance and special relativity (a theory of motion justifying that the speed of light is a sort of asymptote for displacement)
to prevent communication between participants to multi-prover proof systems can be traced back to Kilian\cite{kilian1990strong}.
Probably, the original authors (Ben Or, Goldwasser, Kilian and Wigderson) of \cite{BGKW88} had that in mind already, but it is not explicitly written anywhere.
%Later, \cite{BCMS} identified that the commitments involved in the above proof systems are not suitable for secure implementation of quantum oblivious transfer.
Kent was the first author to venture into {\em sustainable} relativistic commitments \cite{PhysRevLett.83.1447} and introduced the idea of
arbitrarily prolonging their life span by playing some ping-pong protocol between the provers (near the speed of light). This idea was made considerably more practical by
Lunghi \emph{et al.}~in \cite{PhysRevLett.115.030502} who made commitment sustainability much more efficient. 
This culminated into an actual implementation 
by Verbanis \emph{et al.}~in \cite{PhysRevLett.117.140506} where commitments were sustained for more than a day!

As nice as this may sound, such {\em long-lasting} commitments have found so far very little practical use. Consider for instance the zero-knowledge proof for Hamiltonian Cycle
as introduced by Chailloux and Leverrier\cite{CL17}. Proving in Zero-Knowledge that a 500-node graph contains a Hamiltonian cycle would require transmitting $250\,000$
bit commitments (each of a couple hundreds of bits in length)
and eventually sustaining them before the verifier can announce his choice of unveiling the whole adjacency matrix or just the Hamiltonian cycle.
For a graph of $|V|$ vertices, this would require an estimated $200 |V|^{2}$ bits of communication before the verifier can announce his choice $chall$ (see Fig.~\ref{spacetime}).
This makes the application prohibitively expensive. If you use a larger graph, you will need more time to commit, leading to more distance to implement the protocol of \cite{CL17}.
Either a huge separation is necessary between the provers (so that one of them
can unveil according to the verifier's choice $chall$ before he finds out the committal information $B$ used by the other prover while the former must commit all the necessary information
before he can find out the verifier's choice $chall$) or we must achieve extreme communication speeds between prover-verifier pairs. This would only be possible by vastly parallelizing
communications between them at high cost.
Modern (expensive) top-of-line communication equipment may reach throughputs of roughly 1Tbits/sec.
A back of the envelope calculation estimates that the distance between the verifiers must be at least 100 km to transmit $250\,000$ commitments at such a rate.

\begin{figure}[h]
\begin{center}
\includegraphics[angle=0, width=1\textwidth]{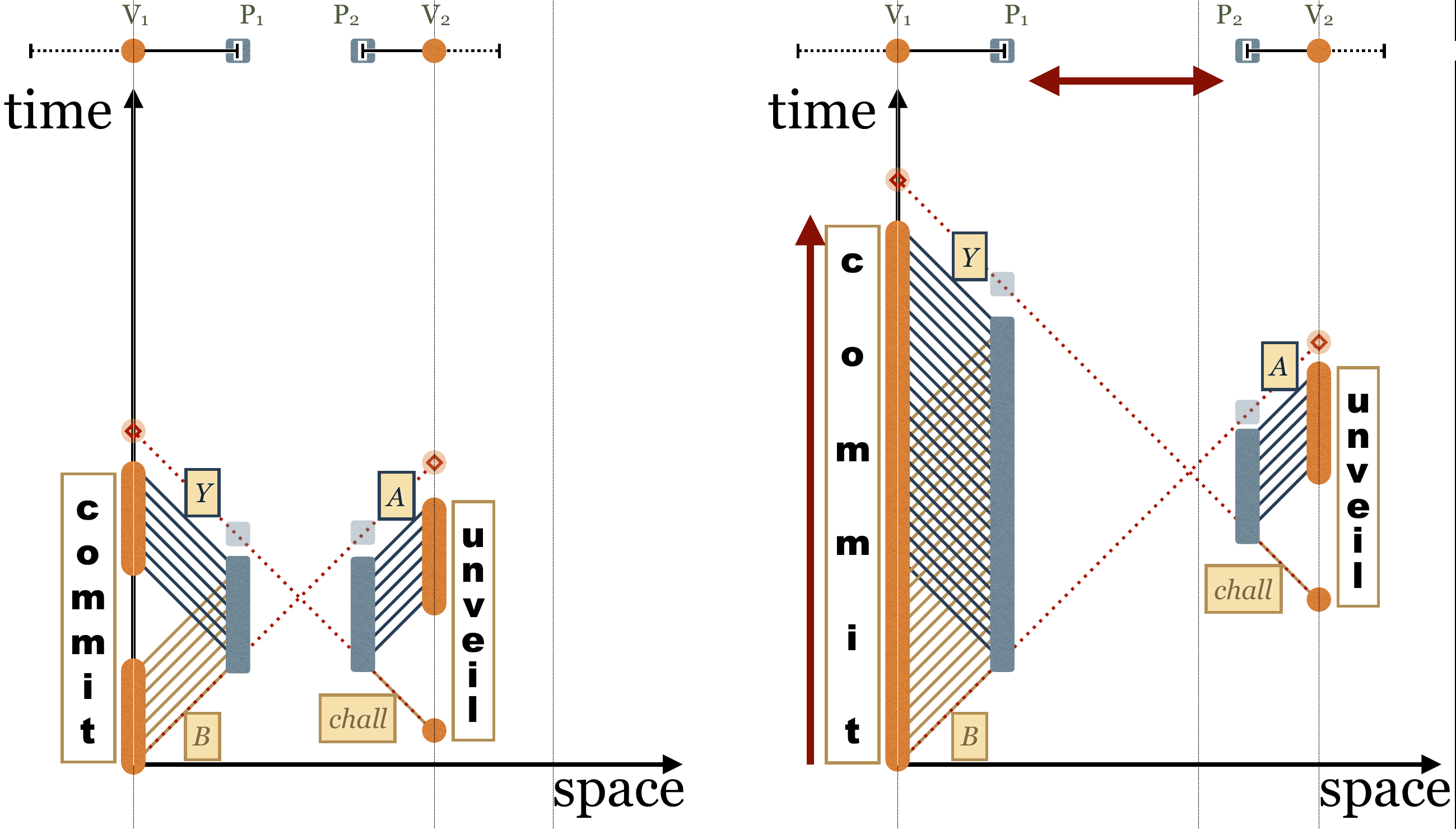}
\end{center}
\caption{Space-Time diagrams of \cite{CL17}'s ZK-MIP${}^{\star}$ for $\mathbf{NP}$. ($45^{\circ}$ diagonals are the speed of light.)} 
\label{spacetime}
In the above two diagrams, $V_{1}$ at a first location sends a random matrix $\fbox{$B$}$ to $P_{1}$
who uses each entry to commit an entry of the adjacency matrix $\fbox{$Y$}$ of $G$. At another location, $V_{2}$ sends a random challenge $\fbox{$chall$}$ to $P_{2}$ who unveils
all or some commitments as $\fbox{$A$}$. At all times, $V_{1}$ and $V_{2}$ must make sure that the answers they get from $P_{1}$ and $P_{2}$ come early enough that the direct communication
line between $V_{1}$ and $V_{2}$ (even at the speed of light) is not crossed.
The transition from left to right shows that increasing the number of nodes (and thus increasing the total commit time) pushes the verifiers further away from each other.
In \cite{CL17} the distance must increase quadratically with the number of nodes in the graph. 
\end{figure}

In this work we consider the following problem: in a Multi-Prover environment, how close can we get the provers in a Zero-Knowledge IP showing the validity of an $\mathbf{NP}$ statement ?
We exhibit a set of (3) novel Zero-Knowledge protocols for the 3-COLorability problem that use two ({\em local}) provers or  three ({\em entangled})  provers and only require them to communicate
two trits each after having each received an edge and two trits each from the verifier.
 This greatly improves the ability to prove Zero-Knowledge statements
on very short distances with very little equipment.
In comparison, the protocol of \cite{CL17} would require transmitting millions of bits between a prover and his verifier before the latter may disclose what to unveil or not. This implies the provers would have to be very
far from each other because all of these must reach the verifier {\em before} the former can communicate with its partner prover. 

%%% Protocols are just too complicated to be practical
Although certain algebraic zero-knowledge multi-prover interactive proofs for $\mathbf{NP}$ and 
$\mathbf{NEXP}$ using explicitly no commitments at all have been presented 
before in \cite{DBLP:journals/combinatorica/LapidotS95},
\cite{DBLP:conf/stoc/FeigeK94} (sound against local provers) and \cite{DBLP:journals/eccc/ChiesaFGS18},\cite{DBLP:journals/eccc/GriloSY19}
 (sound against entangled provers), in the local cases making these protocols entanglement sound
is absolutely non-trivial, whereas in the entangled case the multi-round structure and the amount of communication in each round makes implementing the protocol completely impractical as well.
(To their defense, the protocols were not designed to be {\em practical}).

%The results we exhibit in the current work are of two very different natures: On one hand, we construct extremely simple Zero-Knowledge Multi-Prover interactive proofs for 3COL
%that are respectively sound against (two) Local and (three) Entangled provers, and on the other hand, we show that these protocols satisfy a new stronger notion of multi-simulator Zero-Knowledge where the simulators
%may achieve their task while being at relativistic distance from one another. We introduce the idea that because our new protocols may be simulated by No-Signalling simulators, they cannot
%be sound against No-Signalling provers either.

The main technical tool we use in this work is a general Lemma of Kempe, Kobayashi, Matsumoto, Toner, and Vidick\cite{doi:10.1137/090751293} to prove soundness of a three-prover
protocol when the provers are {\em entangled} based on the fact that a two-prover protocol version is sound when the provers are only {\em local}. More precisely, they proved this when the three-prover
version is the same as the two-prover version but augmented with an extra prover who is asked exactly the same questions as one of the other two at random and is expected to give the same exact answers.

Our protocols build on top of the earlier protocol due to Cleve, H{\o}yer, Toner and Watrous\cite{CHTW04} who presented an extremely simple and efficient solution to
the 3-COL problem that uses only two provers, each of which is queried with either a node from a common edge, or twice the same node. In the former case, the verifier checks that
the two ends of the selected edge are of distinct colours, while in the latter case, he check only that the provers answer the same colour given the same node.
On the bright side, their protocol did not use commitments at all but unfortunately it did not provide Zero-Knowledge either. Moreover, it is a well established fact that this protocol cannot possibly be sound against entangled
provers, because certain graph families have the property that they are not 3-colourable while having entangled-prover pairs capable of winning the game above with probability one.
This was already known at the time when they introduced their protocol.
The reason this protocol is not zero-knowledge follows from the undesirable fact that
dishonest verifiers can discover the (random) colouring of non-edge pairs of nodes in the graph, revealing if they are of the same colour or not in the provers' colouring.

\begin{figure}[h]
\centering
\includegraphics[angle=0, width=0.6\textwidth]{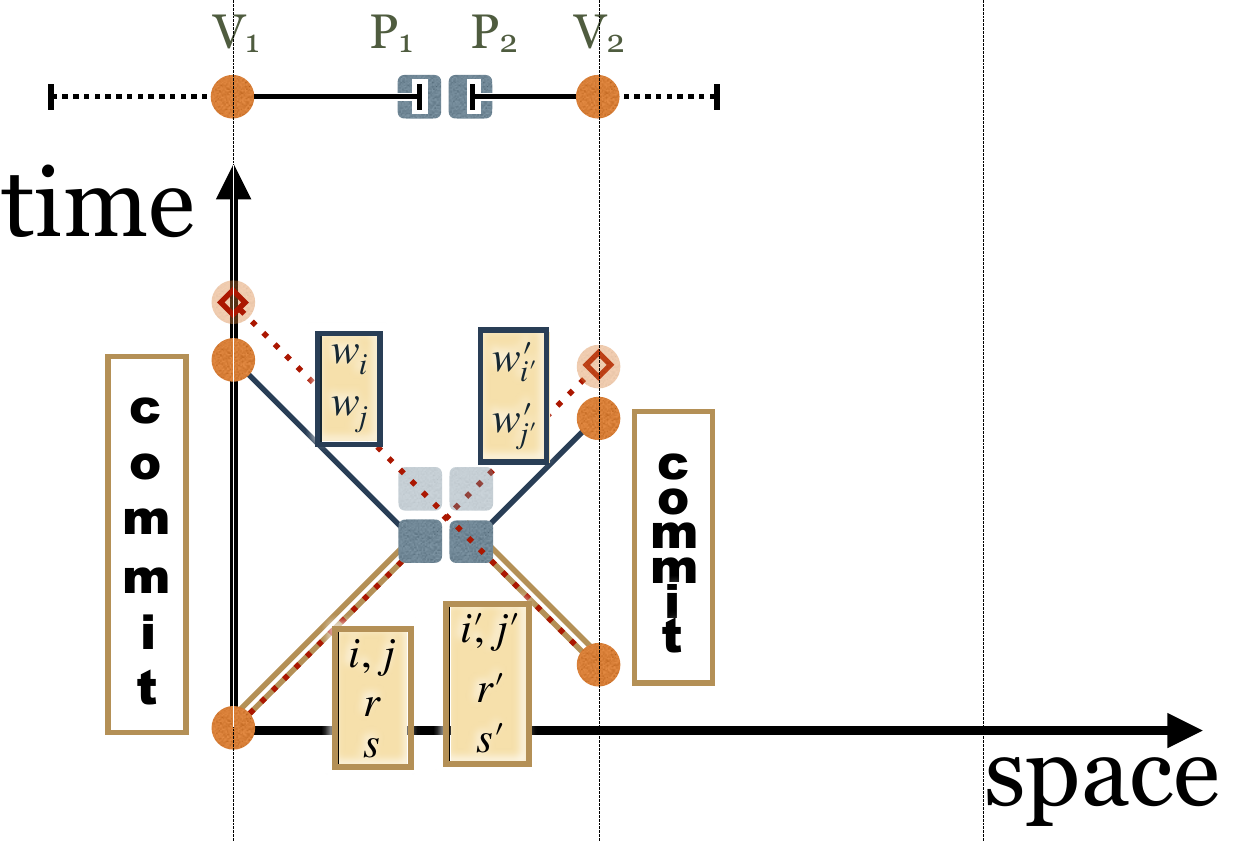}
\caption{Space-Time diagram of our ZK-MIP${}^{\star}$ for $\mathbf{NP}$. ($45^{\circ}$ diagonals are the speed of light. )}
\label{spacetime2020}
\end{figure}

We are able to remedy to the zero-knowledge difficulty by allowing the provers to use commitments for the colour of their nodes. However they use these commitments in an innovative way that we call the
{\em unveil-via-commit principle}  (of independent interest) explained below.
 For this purpose we use commitments similar to those of Lunghi \emph{et al.}\cite{PhysRevLett.115.030502}
but in their simplest form possible, over the field $\mathbb F_{3}$ (or $\mathbb F_{4}$ if you insist working in binary), and thus with extremely weak binding property but also minimal in communication cost:
a complete execution of the basic protocol transmits exactly two node numbers (using only $\log |V|$ bits each) and two trits from verifiers to provers and two trits back from the provers to verifiers (see Fig.~\ref{spacetime2020}).
This implies that for a fixed communication speed, the minimal distance of the provers in our protocol increases logarithmically with the number of nodes whereas the same parameter grows quadratically in \cite{CL17}.
Nevertheless, this is good enough to obtain a zero-knowledge version of the protocol that remains
sound against {\em local} pairs of provers. The main idea being that the provers will each commit to the colours of two requested nodes only if they form an edge of the graph. To unveil the colour of any node, the verifiers
must request commitment of the {\em same} node by {\em both} provers but using different randomizations. This way the verifiers may compute the colour of a node from the {\em linear system} established by the
two commitments and not by explicitly requesting anyone to unveil. This is the unveil-via-commit principle (very similar to the double-spending mechanism of the untraceable electronic cash of Chaum, Fiat and Naor\cite{Chaum:1990:UEC:88314.88969}).
We then use the Lemma of \cite{doi:10.1137/090751293} to prove soundness of the three-prover
version of this protocol even when the provers are {\em entangled}. A positive side of the protocol of \cite{CL17}, however,   is the fact that only two provers are necessary while we use three.
Zero-Knowledge follows from the fact that only two edge nodes can be unveiled
by requesting the same edge to both provers. Otherwise only a single node may be unveiled.
Finally, we show that even the three-prover version of this protocol retains the zero-knowledge property:
requesting any three edges from the provers may allow the verifiers to unveil the colours of a triangle in the graph but never two end-points that do not form an edge
(going to four provers would however defeat the zero-knowledge aspect).

An actual physical implementation of this protocol is currently being developed in collaboration with Pouriya Alikhani (McGill),
Nicolas Brunner, S\'ebastien Designolle, Weixu Shi,
and Hugo Zbinden (Universit\'e de Gen\`eve).

\subsection{Implementations Issues}
Traditionally in the setup of Multi-Prover Interactive Proofs, there is a single verifier interacting with the many provers.
However, when implementing no-communication via spatial separation (the so called relativistic setting)
it is standard to break the verifier in a number of verifiers equal to the number of provers, each of them interacting
at very short distance from their own prover.
The verifiers can use the timing of the replies of their respective provers to judge their relative distance.
In practice, this means that we can implement MIPs under relativistic assumptions
if the verifier are ``split'' into multiple verifiers, each locally interacting with its corresponding prover.
The verifiers use the distance between {\em themselves} to enforce the impossibility of the provers to communicate:
no message from a verifier can be used to reply to another verifier faster than the speed of light {\em wherever the provers are located}.

Moreover, multi-prover interactive proof systems may have several rounds in addition to several provers. 
In general, protocols with several rounds may cause a treat
to the inherent assumption that the provers are not allowed to communicate
during the protocol's execution. Nevertheless, most of
the existing literature resolves this issue by providing an honest verifier that is {\em non-adaptive}. To simplify this task, most of the protocols
are actually single-round. We stick to these guidelines in this work. Moreover, in  order to prove soundness of our protocols against entangled
provers, we use a theorem that is currently only proven for single-round protocols.
The protocols we describe are indeed single-round and non-adaptive.

%\input{3COL-2020_s02_previous_work}

%\input{3COL-2020_s03_standard_model}

%\input{3COL-2020_s04_new_model}

%\input{3COL-2020_s05_local-ZKSCMIP}

%!TEX root = 3COL-2020.tex
\newcommand{\nneq}[0]{\neq\!\!\!\!\neq}
\newcommand{\inn}[0]{\!\in\!}
\newcommand{\protname}[2]{\ensuremath{\Pi^{(#1)}_{\text{#2}}}}
\newcommand{\twopstd}{\protname{2}{\text{std}}}
\newcommand{\twopcl}{\protname{2}{\text{loc}}} % J'ai compris c=classique et q=quantique...
\newcommand{\threepq}{\protname{3}{\text{qnl}}} % Je propose loc=local et qnl=quantum non-local...
\newcommand{\twothreep}{\protname{2/3}{\text{q}}}
\newcommand{\badp}[1]{\ensuremath{\widetilde{#1}}}
\newcommand{\louismod}[2]{#2}
\section{Preliminaries}\label{SEC:PRELIM}

\subsection{Notations}
Random variables $A,B\in \Gamma$ are said to be equivalent, denoted $A=B$,
if for all $x\in \Gamma$, $\Pr{(A=x)}=\Pr{(B=x)}$.
The class of probabilistic polynomial-time  Turing machines 
will be denoted $\ppt$ in the following. A \ppt\ Turing machine
is one having access to a fresh infinite read-only tape of random values
(uniform values from the set of input symbols) 
at the outset of the computation. In the following, adversaries
will also be allowed (in some cases) to be quantum machines. 
The precise ways quantum and classical machines are defined 
is  not important in the following.  

For $M$ a Turing machine, we denote by $M(x)$ 
it execution with $x$ on its input tape ($x$ being a string of 
the tape alphabet symbols).
A Turing machine (quantum or classical) augmented with read-only auxiliary-input tapes   
and write-only auxiliary-output tapes is called  an \emph{interactive Turing machine}
(ITM).
Read-only input tapes provide incoming messages while the write-only output
tapes  allow to send messages. Interactive Turing machine $M_1$ and $M_2$ 
are said to \emph{interact} when for each of them, 
one of its write-only auxiliary-output tape corresponds
to one read-only auxiliary-input tape of the other Turing machine.
An execution of 
interactive Turing machines $M_1,\ldots, M_k$ on common input $x$
is denoted $[M_1\ldots M_k](x)$.  For $1\leq i \leq k$, 
machine $M_i$ \emph{accepts} the interactive computation
on input $x$ if it stops in state $\mathsf{accept}$ 
after the execution $[M_1\ldots M_k](x)$.
When the ITM $M_i$ that accepts a computation is clear from the context,
we say that $[M_1\ldots M_k](x)$ accepts when $M_i$'s
final state is $\mathsf{accept}$. 
In this scenario, $\Pr{\left([M_1\ldots M_k](x)=\mathsf{accept}\right)}$
denotes the probability that $M_i$ terminates in state $\mathsf{accept}$
upon common input $x$. Quantum machines are also interacting 
through communication tapes the same way than for classical machines.
When a quantum machine $M_1$ interacts with a classical machine 
$M_2$, we suppose that the write-only auxiliary tape and the reade-only
auxiliary tape of $M_1$ used to communicate with $M_2$ are classical.
This is the situation we will be addressing almost all the time in the 
following. A quantum machine $M$ is also allowed to have a quantum auxiliary
read-only input tape that may contain a part of a quantum state shared with other
machines. This allows to model machines sharing entanglement at the outset
of an interactive computation.  
Henceforth, we suppose that the (main) input tape of all machines
(quantum or classical)  is classical.

In the following, $G=(V,E)$ denotes an undirected graph
with vertices $V$ and edges $E$. If $n=|V|$ then we denote 
the set of vertices in $G$ by $V=\{1,2,\ldots, n\}$. We suppose
that $(i,i)\notin E$ for all $1\leq i \leq n$ (i.e. $G$ has no loop).
We denote uniquely each edge in $E$ as $(i,j)$ with $j>i$.
For $i\in V$, let $\mathsf{Edges}(i) := \{(j,i)\in E\}_{j<i}\cup \{(i,j)\in E\}_{j>i}$ be 
the set of edges connecting vertex $i$ in $G$.
For $e,e'\in E$, we define $e\cap e'=i\in V$ if $e$ and $e'$ have only
one vertex $i\in V$ in common. When $e$ and $e'$ have four distinct
vertices in $V$, we set $e\cap e' =0$. Finally, when $e=e'$, 
we set $e\cap e' := \infty$. 
For readability, we use the following special notations: $(a, b) \nneq (c,d)$ means $a\neq c$ {\bf and} $b\neq d$, while as always, $(a, b) \neq (c,d)$ simply means $a\neq c$ {\bf or} $b\neq d$.

\subsection{Non-local Games, Multi-Prover Interactive Proofs, and Relativistic Proofs}

Multi-provers interactive protocols are protocols involving a set of \emph{provers}
modelled by interactive Turing machines, each of them interacting
with an interactive \ppt\ Turing machine called the verifier $\verifier$. 
Although all provers may share an infinite read-only auxiliary input tape
at the outset of their computation, they do not not interact with
each other. When the provers are quantum, an extra auxiliary read-only quantum input 
tape is given and can be entangled with other provers at the beginning.

\begin{definition}

Let $\prover_1, \ldots, \prover_k$ be computationally unbounded interactive
Turing machines and let $V$ be an interactive \ppt\ Turing machine. 
The $\prover_i$'s share a joint, infinitely long, read-only random tape
(and an auxiliary reads-only quantum input tape if the provers are quantum). 
Each $\prover_i$ interacts
 with $\verifier$ but cannot interact with $\prover_j$ 
 for any $1\leq j\neq i\leq k$. We call $[\prover_1, \ldots, \prover_k, \verifier]$ 
 a \emph{$k$-prover interactive protocol ($k$--prover IP)}.

\end{definition}

A $[\prover_1, \ldots, \prover_k, \verifier]$ 
$k$-prover interactive protocol is a \emph{multi-prover interactive proof system for $L$}
if it can be used to show $\verifier$ that 
a public input $x$ is such that $x\in L$. At the end of its computation, 
$\verifier$ concludes $x\in L$ if and only if it ends up in state $\textsf{accept}$.
We restrict our attention to interactive proof systems with perfect
completeness since all our protocols have this property.

%We write $[\prover_1,\ldots, \prover_k, \verifier](x)$ for the random variable
%the final state reached by $\verifier$ after interacting with $\prover_1,\ldots, \prover_k$
%upon public input $x$. The probability that $\verifier$ accepts in those
%conditions is simply denoted by $\Pr{\big([\prover_1,\ldots, \prover_k, \verifier](x) = \textsf{accept}\big)}$.
\begin{definition} \label{kps}
The $k$--prover interactive protocol
$\Pi=(\prover_1, \ldots, \prover_k, \verifier)$
is said to be a \emph{$k$-prover interactive proof system
with perfect completeness for $L$} if there exists 
$q(n) < 1-\frac{1}{\poly}$ such that following holds:
\begin{description}
\item[perfect completeness:] $(\forall x \in L)\left[ \Pr{\big([\prover_1,\ldots, \prover_k, \verifier](x)= \emph{\textsf{accept}}\big)} = 1\right]$,
\item[soundness:] $(\forall x \notin L)
(\forall \badp{\prover}_1, \ldots, \badp{\prover}_k)\left[ 
\Pr{\big([\badp{\prover}_1,\ldots, \badp{\prover}_k, \verifier](x)= \emph{\textsf{accept}}\big)}\leq q(|x|)\right]$.
\end{description}
The parameter $q(|x|)$ is called the \emph{soundness error} of $\Pi$.
Soundness can hold against classical provers or against
quantum provers sharing entanglements. The former case is called
\emph{sound against classical provers} while to latter is called 
\emph{sound against entangled provers}.
\end{definition}

Consider a $k$--prover interactive proof system 
$\Pi(x)$ (with or without perfect
completeness) for $L$ executed with public input $x \notin L$.
In this situation, $\Pi(x)$ defines what is called a \emph{quantum
game}.
The minimum value $q(|x|)$ such that for all $\prover'_1,\ldots, \prover'_k$,
$\Pr{\big([\prover'_1,\ldots, \prover'_k, \verifier](x)= \emph{\textsf{accept}}\big)}\leq q(|x|)$
is often called the \emph{classical value of  game} $\Pi[x]$ and is
denoted $\omega(\Pi(x))$ when the provers are restricted to be classical and unable
to communicate with each other upon public input $x$. When the provers, still unable 
to communicate with each other,  are allowed to carry their
computation quantumly and share entanglements,
we denote by $\omega^*(\Pi(x))\geq \omega(\Pi(x))$ the minimum value $q(|x|)$ 
such that for all such quantum provers $\prover'_1,\ldots,\prover'_k$,
$\Pr{\big([\prover'_1,\ldots, \prover'_k, \verifier](x)= \emph{\textsf{accept}}\big)}\leq q(|x|)$. In this case, $\omega^*(\Pi(x))$
is called the \emph{quantum value of game} $\Pi(x)$.
A $k$--prover interactive proof system for $L$ is said to be \emph{symmetric}
if \verifier\ can permute the questions to all provers without 
changing their distribution.  
The following result
of  Kempe, Kobayashi, Matsumoto, Toner, and Vidick\cite{doi:10.1137/090751293}
shows that the classical value of a symmetric one-round classical game 
cannot be too far from the quantum value of a \emph{modified} game.
Given a symmetric one-round two-prover game $\Pi$,
one can always add a third prover $\prover_3$  and \verifier\ asks
$\prover_3$ the same question than $\prover_1$ with probability $\frac{1}{2}$
or the same question than $\prover_2$  with probability $\frac{1}{2}$. Then,
\verifier\ accepts if $\prover_1$ and $\prover_2$ would be accepted in $\Pi(x)$
and   if $\prover_3$ returns the same answer than the one returned 
by the prover it emulates. We call $\Pi'(x)$ the modified game 
obtained that way from $\Pi(x)$.

\begin{lemma}[\cite{doi:10.1137/090751293}, Lemma~17]\label{values}
Let $\Pi(x)$ be a two-prover one-round symmetric game and
let $\Pi'(x)$ be its modified version with three provers. 
If $\omega^*(\Pi'(x))>1-\varepsilon$ 
then $\omega(\Pi(x))>1- \varepsilon - 12|Q|\sqrt{\varepsilon}$ where
$Q$ is the set of \verifier's possible questions to a prover in $\Pi$. 
\end{lemma}
%\begin{figure}[h]
%\centering
%\includegraphics[angle=0, width=0.8\textwidth]{KKMTV09.png}
%\caption{Dixit KKMTV-09}
%\label{KKMTV09}
%\end{figure}

Lemma~\ref{values} remains true for non-symmetric two-prover one-round
protocol by  first making them symmetric at the cost of increasing
the size of $Q$. This is always possible
without changing the classical value of the game and by using twice the
number of questions $|Q|$ of the original game (Lemma~4 in \cite{doi:10.1137/090751293}).

Let $[\prover_1, \ldots, \prover_k, \verifier]$  be a $k$--prover
IP.  We denote by $\mathbf{view}(\prover_1, \ldots, \prover_k, \verifier, x)$ 
the probability distribution of \verifier's outgoing and incoming
messages with all provers according $\verifier$'s
coin tosses. 

%Zero-knowledge MIPs were also defined in \cite{BGKW88}:

\begin{definition}\label{DEF:standardZK}

Let $[\prover_1, \ldots, \prover_k, \verifier]$ be a $k$-prover interactive
proof system for $L$. 
We say that $[\prover_1, \ldots, \prover_k, \verifier]$ 
is \emph{perfect zero-knowledge}  if 
for all \ppt\ interactive Turing machines 
$\badp{\verifier}$
there exists a \ppt\
machine \simulator\ (i.e. the \emph{simulator}) having blackbox
access to $\badp{\verifier}$ such that  for all $x$,
\[\mathbf{view}(\prover_1, \ldots, \prover_k, \badp{\verifier}, x) = \simulator(x)\enspace,\]
and both random variables are equivalent.
In the following, we allow $\badp{\verifier}$ to be a quantum machine 
but our simulators will always be classical machines with blackbox access to $\badp{\verifier}$.
If the zero-knowledge condition holds against quantum $\badp{\verifier}$,
we say that the proof system is \emph{perfect zero-knowledge against quantum verifiers}. 
\end{definition}

%\LC{Est-ce \`a la bonne place:iu
%There have also been augmentations of the model by giving the {provers} various non-local resources, such as entanglement \cite{IV12}, or arbitrary no-signaling power \cite{KRR14}.} 

%\begin{definition}
%
%An \emph{interactive Turning machine} (ITM) is a Turing machine augmented with the following tapes:
%
%\begin{itemize}
%
%\item $k_1$ read-only incoming communication tapes.
%
%\item $k_2$ write-only outgoing communication tapes.
%
%\item Private work, auxiliary-input, and random tapes.
%
%\end{itemize}
%
%An ITM $A$ can signal to an ITM $B$ if $A$'s write-only outgoing tape is $B$'s read-only incoming tape.
%
%\end{definition}

%\subsection{Zero-Knowledge Interactive Proofs for $\mathbf{3}$-$\mathbf{COL}$}

\subsection{Multi-Prover Commitments with Implicit Unveiling}\label{commitment}
Our multi-prover proof systems for \textsf{3COL}
use a simple 2-committer commitment scheme with a
property allowing to guarantee perfect zero-knowledge.
In this section, we give the description of this simple commitment
scheme with its important  properties four our purposes.

Assume that provers $\prover_1$ and $\prover_2$ share 
$\ell$  values $c_1,c_2,\ldots,c_{\ell}\in \mathbb{F}$
where $\mathbb{F}$ is a finite set.
$\verifier$ wants to check that these values satisfy some properties
without revealing them all. Assume that $\mathbb{F}$ 
is a field with operations $+$ and $\cdot$.

Bit commitment schemes have been used in the multi-prover model ever since it was
 introduced in \cite{BGKW88}.
The original scheme was basically $w_{i} := b_{{i}} \cdot r_{i} + c_{{i}}$, 
a commitment $w_{i}$ to value $c_{{i}}\in \mathbb{F}$ 
using pre-agreed random mask $b_{{i}}\in_R \mathbb{F}$
and randomness ${r_{i}}\neq 0$ provided by \verifier.
Kilian\cite{KILIAN89} had a binary version where each bit $c_i:= c^{1}_i \oplus c^{2}_i \oplus c^{3}_i$ is shared among provers $\prover_{1}$ and $\prover_{2}$
(and therefore $\mathbb{F}$ needs only to be a group).
To commit $c_i$, \verifier\ samples  $c^{h}_i$ from $\prover_{1}$ and 
$c^{j}_i$ from $\prover_{2}$ at random.
If $j=h$ but $c^j_{i}\neq c^h_{i}$, 
\verifier\ immediately rejects the commitment. 
Otherwise either $\prover_{1}$ or $\prover_{2}$ may unveil 
by disclosing $c^{1}_i,c^{2}_i,c^{3}_i$ at a later time.
%Lapidot-Shamir somewhat do not use commitments. The amount of communication in each round is roughly 0( |V|^{2} ) bits.
Somehow, bad recollection of \cite{BGKW88}'s scheme lead \cite{BCMS} to a similar but different scheme defining $w_{i} := c_{{i}} \cdot r_{i} + b_{{i}}$, 
a commitment $w_{i}$ to bit $c_{{i}}\in \{0,1\}$
using pre-agreed bit mask $b_{{i}}\in_R \{0,1\}$  and binary randomness 
${r_{i}}$ provided by their corresponding verifiers.
Although this form of commitment is  intimately  connected to the 
CHSH game \cite{CHSH} and the Popescu-Rohrlich box\cite{PR94},
this proximity is not relevant for the soundness and the completeness
of our  protocols, even against entangled provers. Although the 
(limited) binding property of these schemes has been established
in  \cite{PhysRevLett.83.1447,CSST11,PhysRevLett.117.140506,FF15,PhysRevLett.115.030502,CL17} against entangled provers,
we only use this commitment scheme against classical provers,
only sharing classical information before the execution
of the protocol. The weak binding property of  these schemes
against entangled provers
does not allow us to get sound and complete proof systems
against these provers.  We shall rather get completeness and
soundness against entangled provers 
using a different technique 
from \cite{doi:10.1137/090751293} that requires a third prover.

For an arbitrary field $\mathbb{F}$,  the commitment scheme produces
commitment
$w_{i} := c_{{i}} \cdot r_{i} + b_{{i}}$ to field element $c_{{i}}\in\mathbb{F}$
using pre-agreed field element mask $b_{{i}}$ (specific to value $1\leq i\leq \ell$) 
and random field element ${r_{i}}\neq 0$ provided by their corresponding verifiers.
Many results were proven for this specific form of the commitments. 
Notice however that the two versions discussed above,
$w_{i} := b_{{i}} \cdot r_{i} + c_{{i}}$ in the former case and $w_{i} := c_{{i}} \cdot r_{i} + b_{{i}}$ in the latter have equivalent binding property(left as a simple exercice).
Considering, the former as being the degree-one secret sharing \cite{DBLP:journals/cacm/Shamir79} of $c_{{i}}$ hidden in the degree zero term,
while the latter being the degree-one secret sharing of $c_{{i}}$ hidden in the degree one term,
we decided to use the former (original BGKW form) because all the known results about secret sharing are generally presented in this form. In particular, this form is more adapted to higher degree generalizations such as
$w_{i} := a_{{i}} \cdot r_{i}^{2} + b_{{i}} \cdot r_{i} + c_{{i}}$ being the degree-two secret sharing of $c_{{i}}$ hidden in the degree zero term, and so on.

Moreover, this choice turns out to simplify our (perfect) zero-knowledge simulator.
For the rest of this paper, we use $w_{i} := b_{{i}} \cdot r_{i} + c_{{i}}$
where $w_i, b_i, c_i\in \mathbb{F}_3$ and $r_i\in \mathbb{F}_3^*$. 
Provers therefore commit to trits,
one value for each node corresponding to its
colour in a $3$--colouring of graph $G=(V,E)$. The 
values shared between $\prover_1$ and $\prover_2$
are therefore, for each node $i\in V$, the colour $c_i$ of that node.

Suppose that \verifier\ asks $\prover_1$ to commit on the colour $c_i$
of node $i\in V$ using randomness $r \in_R \mathbb{F}^*_3$.
Let $w= b_i \cdot r + c_i$ be the commitment returned to \verifier\
 by $\prover_1$. Suppose \verifier\ asks $\prover_2$ to commit
 on the colour $c'_j$ of node $j\in V$ using randomness $r' \in _R\mathbb{F}^*_3$.
 Let $w' = b_j \cdot r' + c'_j$ be the commitment 
 issued to \verifier\ by $\prover_2$. The following
 3 cases are possible depending on \verifier's choices
 for $i,j,r$, and $r'$:
 \begin{enumerate}
\item \emph{(forever hiding)}  if $i\neq j$ then \verifier\ learns nothing on neither $c_i$ nor $c'_j$
 since $w$ and $w'$ hide $c_i$ and $c'_j$ with random and independent
 masks $b_i\cdot r$ and $b_j\cdot r'$ respectively. Even knowing 
 $r,r'\in \mathbb{F}_3^*$, $b_i\cdot r$ and $b_j\cdot r'$
 are uniformly distributed in $\mathbb{F}_3$.\\[-0.2cm]
\item \emph{(the consistency test)} If $i=j$ and $r=r'$ then \verifier\ can verify
%\CCom{Je ne vois pas la pertinence de ``that $c_i=c'_j$ by  testing''...  On veut surtout s'assurer que les prouveurs sont d\'eterministes.}
that $w=w'$. This corresponds to the immediate rejection of \verifier\ in
Kilian's two-prover commitment described above. It allows \verifier\ to make sure
that $\prover_1$ and $\prover_2$ are consistant when asked to commit
on the same value.\\[-0.2cm]
\item \emph{(implicit unveiling)} If $i=j$ and $r' \neq r$ then \verifier\ can learn 
$c_i$ (assuming $w=b_i\cdot r + c_i$ and $w'=b_i\cdot r' + c_i$) the following way. 
\verifier\ simply computes $c_{i} := 2^{-1}\cdot(w+w')$
(Note that over an arbitrary field $c_{i} := (wr'\!-w'r)(r'\!-r)^{-1}$ whenever $r\neq r'$).  Interpreting the meaning 
of this test can be done when considering a strategy for $\prover_1$ 
and $\prover_2$ that always passes the consistency test. 
In this case, $w=b_i\cdot r + c_i$ and
$w'=b_i\cdot r' + c_i$ are satisfied and \verifier\ learns the committed value $c_i$.  
\end{enumerate}
As long as $\prover_1$ and $\prover_2$ are  {\em local} (or
{\em quantum non-local})
they cannot distinguish which option \verifier\ has picked among the three.
The consistency test makes sure that if $\prover_1$ and $\prover_2$
do not commit on identical values for  
some $1\leq i \leq \ell$ then \verifier\ 
will detect it when \verifier\ picks the  consistency test
for commitment  $w$ and $w'$  in position $i$.

%\LC{Je pense que ceci devrait \^etre d\'eplac\'e:
%We assume all $b_{{i}}$'s uniform and independent. 
%Indeed, our protocols use at most six of them and therefore as long as $f$ is
%sampled from a 6-universal function family $F$, $b_{n_{i}} := f(n_{i})$ are good enough for our purpose. 
%
%In an actual implementation of our protocols, it is probably convenient to prepare everything ready before the provers split from each other.
%An array $W[n_{i},n_{j},r_{i},r_{j}]$ of size $|E|\times |\mathbb{F}|^{2}$ for all relevant future values of $w$ should be pre-computed by the provers.
%At run-time, their only role will be to look up two values each in that table, and thus making the protocol extremely efficient.}

\section{Classical Two-Prover Protocol}\label{SEC:PROTs}
First, consider a small variation over the protocol of Cleve et al.
presented in \cite{CHTW04}.
In their protocol, when $\prover_1$ and $\prover_2$ both know 
and act upon the same valid 3-colouring of $G$, \verifier\ asks each prover for the colour of a vertex in $G=(V,E)$.
Consistency is verified when \verifier\ asks the same vertex to each prover 
and compares that the same colour has been provided. The colorability is checked 
when the provers are asked for the colour of two connected vertices in $G$. 
This way of proceeding is however problematic for the  zero-knowledge 
condition. \verifier\ could be asking two nodes that do not form an edge
for which their respective colour will be unveiled. This  
certainly allows \verifier\ to learn something about $\prover_1$'s and $\prover_2$'s
colouring. Indeed, repeating this many times will allow  \verifier\ to efficiently reconstruct a complete colouring.
To remedy partially this problem, \verifier\ is instead 
asking each prover the colouring of an entire edge of $G$.
The colouring is (only) checked when both provers are asked
the same edge, while consistency
is checked when two intersecting edges are asked to the provers.

\subsection{Distribution of  questions}\label{distquestions}
Let $G=(V,E)$ be a connected undirected graph.
Let us define the probability distribution
$\mathcal{D}_G=\{(p(e,e'),(e, e'))\}_{e,e'\in E}$ for the pair $(e, e')\in E\times E$
that \verifier\ picks with probability $p(e,e')$ before announcing $e$ to $\prover_1$ and 
$e'$ to $\prover_2$.  
For $e,e'\in E$ such that 
$e\cap e'=0$, we set $p(e,e') :=0$ so that \verifier\ never asks 
two disconnected edges in $G$ (this would give no useful information).

The first thing to do is to pick
$e=(i,j)\in E$ uniformly at random. With probability $\epsilon$ (to be selected later),
we set $e'=e$, which allows  for an edge-verification test. 
With probability $1-\epsilon$, we perform a well-definition test
as follows.
With probability $\frac{1}{2}$, $e'\in \text{\sf{Edges}}(i)$ 
uniformly at random and with probability $\frac{1}{2}$, 
 $e'\in \text{\sf{Edges}}(j)$ 
uniformly at random. 
In other words, the well-definition test picks
the second edge $e'$ with probability $\frac{1}{2}$
among the edges connecting $i\in V$ and with probability $\frac{1}{2}$
among the edges connecting $j\in V$.
%%we set $p(e,e):=\frac{1}{2|E|}$ for all $e\in E$
%%and \verifier\ asks the same edge twice with probability $\frac{1}{2}$.
%Finally, with probability $\frac{1}{2}$,
%$e=(i,j)\in_R E$ is picked randomly and  the second edge $e'$ 
%will be picked uniformly in $\mathsf{Edges}(i)$ with probability
%$\frac{1}{2}$ and picked uniformly in $\mathsf{Edges}(j)$ 
%with probability $\frac{1}{2}$.
It follows that  for $e' \in \mathsf{Edges}(i)\cup \mathsf{Edges}(j)$ with $e\neq e'$,
%%$p(e,e')= \frac{1}{4|E|\cdot |\mathsf{Edges}(i)|}\geq \frac{1}{4\cdot|E|\cdot|V|}$,
%%and for $e' \in \mathsf{Edges}(j)$ with $e\neq e'$,
%%$p(e,e')=\frac{1}{4|E|\cdot |\mathsf{Edges}(j)|}\geq \frac{1}{4\cdot|E|\cdot|V|}$.
%%It follows that for $e,e'\in E$ such that $e\cap e' \notin \{0,\infty\}$,
we have, for $e=(i,j)\in E$,
\begin{equation}\label{bpee}
p(e,e') = \frac{1-\epsilon}{2|E|}\left( \frac{|\{e'\}\cap \text{\sf{Edges}}(i)|}{|\text{\sf{Edges}}(i)|}+\frac{|\{e'\}\cap \text{\sf{Edges}}(j)|}{|\text{\sf{Edges}}(j)|}\right)\enspace. %\geq \frac{1}{4\cdot|E|\cdot|V|}\enspace.
\end{equation} 
We also get
\begin{equation}\label{bpee2}
p(e,e)= \frac{\epsilon}{|E|} + 
\frac{1-\epsilon}{2|E|}\left( \frac{1}{|\text{\sf{Edges}}(i)|}+\frac{1}{|\text{\sf{Edges}}(j)|}\right) \geq \frac{\epsilon}{|E|}\enspace.
\end{equation}
%For $e\in E$, 
%$e'$ is picked 
%uniformly in $\mathsf{Edge}(i)-\{e\}$ (in which case $e\cap e'=i$)
%with probability $\frac{1}{2}$.
%Likewise, 
%for $e=(i,j)\in E$ and $e'=(h,j)$ or $e'=(j,h)$ with $h\neq i$,
%$p(e,e')=\frac{1}{2|E|\cdot (|\mathsf{Edge}(j)|-1)}$.
%For $e\in E$, $e'$, is picked 
%uniformly in $\mathsf{Edge}(i)-\{e\}$ (in which case $e\cap e'=j$)
%also with probability $\frac{1}{2}$. It means that 
%for $e=(i,j)$, $e'$ is picked uniformly at random such that $e\cap e'=i$  
%with probability $\frac{1}{2}$ and 
%picked uniformly at random such that $e\cap e'=j$   
%with probability $\frac{1}{2}$.  
%Finally,
%for $e=(i,j)\in E$, $p(e,e)= \frac{1}{2|E|\cdot |\mathsf{Edge}(i)|} + \frac{1}{2|E|\cdot|\mathsf{Edge}(j)|}$.
%In other words, \verifier\ starts by picking $e=(i,j)\in_R E$ uniformly 
%at random and then picks $e'$ uniformly at random in $\mathsf{Edge}(i)$
%with probability $\frac{1}{2}$ and uniformly at random in $\mathsf{Edge}(j)$
%also with probability $\frac{1}{2}$.

It is easy to verify that $\mathcal{D}_G$
is a properly defined probability distribution over pairs
of edges.% in $G$.

\subsection{A Variant Over the Two-Prover Protocol of Cleve et al.}
Distribution $\mathcal{D}_G$ produces two edges where
the first one is provided to $\prover_1$ while the second 
one is provided to $\prover_2$. Each prover then returns
the colour of each node of the edge to \verifier.
We denote the resulting protocol
\twopstd.\\\\

\begin{quote}\label{MIP_3COL}
\rule{\linewidth}{1pt}
\prot{\,\,{\twopstd}[G]}{\,\,Two-prover, 3-COL.}
Provers $\prover_{1},\prover_{2}$ 
pre-agree on a random 3-colouring of $G$: $\left\{ (i,c_{i}) | c_i \in \
\mathbb{F}_3 \right\}_{i\in V}$ 
such that $({i},{j}) \inn E \implies c_{{j}} \neq c_{{i}}$.
%The edges $(n_{0},n_{1}),(n_{2},n_{3})$ are selected under one of two constraints: either
%$$(n_{0},n_{1}) = (n_{2},n_{3}) \text{ or } \exists (i,j) \inn \{0,1\}\times \{2,3\} : n_{i}=n_{j}, \text{ but } (n_{0},n_{1}) \neq (n_{2},n_{3}).$$

{\bf Interrogation phase:}

\begin{itemize}

\item $\verifier$ picks $((i,j), (i'\!,j'))\in_{\mathcal{D}_G} E\times E$, sends
 $(i,j)$ to $\prover_1$ and $(i'\!,j')$ to $\prover_2$.
\item If $(i,j) \inn E$ then $\prover_1$ replies with $c_{i},c_{j}$.
\item If $(i'\!,j') \inn E$ then $\prover_2$ replies with $c_{i'},c_{j'}$.

\end{itemize}

{\bf Check phase:}

\begin{itemize}
\item 
{\bf Edge-Verification Test:}\\ if $(i,j) = (i',j')$ then $\verifier$ accepts iff  $c_{i} = c_{i'} \neq c_{j'} = c_{j}$.

\item
{\bf Well-Definition Test:}\\ if $(i,j)\cap (i',j')=h\in V$  then 
\verifier\ accepts iff $c_{h} = c'_{h}$.

\end{itemize}

\rule{\linewidth}{1pt}
\end{quote}

The perfect soundness of this protocol is not difficult to establish along
the same lines of the proof of soundness for the original protocol in \cite{CHTW04}. 
On the other hand,
zero-knowledge does not even hold against honest verifiers. 
\verifier\ learns the colour of each node  
contained in any two  edges of $G$. This is certainly information
about the colouring that \verifier\ learns after the interaction.
To some extend, the modifications we applied to the 
2-prover interactive proof system of \cite{CHTW04}
leaks even more to \verifier. In the next section, we show
that the 2-prover commitment scheme, that we introduced in 
Sect.~\ref{commitment}, can be used 
in protocol $\twopstd$  to prevent this leakage completely.

\section{Perfect Zero-Knowledge Two-Prover Protocol}
\label{MIP_3COLZK}

We modify the protocol of section \ref{MIP_3COL} to prevent \verifier\ 
from learning the colours of more than two connected
nodes in $G$. The idea is simple, $\prover_1$ and
$\prover_2$ will return commitments for the colours
of the nodes asked by \verifier. 
The implicit unveiling of the commitment
scheme described in section \ref{commitment}
will allow \verifier\ to perform both the edge-verification
and well-definition tests in a very similar way that
in protocol $\Pi^{(2)}_{\text{std}}$. The commitments
require \verifier\ to provide a random nonzero trit for each
node of the edge requested to a prover. 

\subsection{Distribution of questions}
We now
define the probability distribution 
$\mathcal{D}'_G$ for  \verifier's questions 
in protocol $\twopcl[G]$ defined in the following section.
It  consists in one edge and two nonzero trits for each prover:
\[\mathcal{D}'_G=\{(p'(e,r,s,e'\!,r'\!,s'),((e,r,s), (e'\!,r'\!,s'))\}_{e,e'\in E, 
r,s,r'\!,s'\in\mathbb{F}^*_3}\]
upon graph $G=(V,E)$ and where $(e,r,s)$ is the question to $\prover_1$
and $(e'\!,r'\!,s')$ is the question to $\prover_2$. 
$\mathcal{D}'_G$ is easily derived from the distribution 
$\mathcal{D}_G=\{(p(e,e'),(e,e'))\}_{e,e'\in E}$ for the questions in ${\twopstd}[G]$,
as defined in section \ref{distquestions}.
First, an edge $e\in_R E$
is picked uniformly at random. Together with $e$, two nonzero trits $r,s\in_R \mathbb{F}^*_3$ 
are picked at random. Then, as in $\mathcal{D}_G$, with probability ${\epsilon}$ (to be selected later) the second 
edge $e'=e$, in which case
we always set $r'=-r$ and $s'=-s$.
This case allows 
for an edge-verification test. 
Finally, with probability $1-\epsilon$,
we pick $e'$ with probability $p(e,e')$  and 
pick $r'\!,s'\in_R\mathbb{F}^*_3$ so that
the couple $((e,r,s),(e'\!,r'\!,s'))$ is produced with probability
$\frac{1}{16}p(e,e')$
for all $e,e'\in E$, and $r,s,r'\!,s'\in \mathbb {F}^*_3$.
This will allow for a well-definition test.
A consequence of (\ref{bpee}) is that for $e=(i,j)\in E$,
$e'\in  \mathsf{Edges}(i)\cup  \mathsf{Edges}(j)$ with $e\neq e'$,
\begin{equation}\label{bbpee}
p'(e,r,s,e'\!,r'\!,s') \geq\frac{1-\epsilon}{16|E|}\left( \frac{|\{e'\}\cap \text{\sf{Edges}}(i)|}{|\text{\sf{Edges}}(i)|}+\frac{|\{e'\}\cap \text{\sf{Edges}}(j)|}{|\text{\sf{Edges}}(j)|}\right)\enspace.
\end{equation}
According to (\ref{bpee2}), we also get  
\begin{equation}\label{bbpee2}
p'(e,r,s,e,r,s) = \frac{p(e,e)}{4} \geq \frac{\epsilon}{4|E|}\enspace.
\end{equation}
%The case $e=e'$ with randomness 
%$r,s, r'=-r$, and $s'=-s$ allows \verifier\ to perform the edge-verification test.
%When in this case the randomness is such that $r'=r$ or $s'=s$ then
%\verifier\ performs the well-definition test.
%For $e=(i,j)\neq e'=(i',j')$
%with $i=i'$, we have $p((e,r,s),(e',r,s'))=\frac{1}{8}p(e,e')$, with $i=j'$
%we have $p'((e,r,s),(e',r',r))=\frac{1}{8}p(e,e')$, with $j=i'$ we have
% $p'((e,r,s),(e',s,s'))=\frac{1}{8}p(e,e')$, and 
% with $j=j'$ we have  $p'((e,r,s),(e',r',s))=\frac{1}{8}p(e,e')$.
% It means that when $\verifier$ asks for two distinct edges,
% the common node will be announced with the same randomness.
% The case $e\neq e'$ allows \verifier\ to run the well-definition test. 
 It is easy to verify that $\mathcal{D}'_{G}$ is a properly defined probability
 distribution.  

 \subsection{The  Protocol}
The protocol is similar to $\twopstd$ except that instead of returning
to \verifier\ the colour for each node of an edge in $G$,
each prover returns commitments with implicit unveilings of these colours.
If \verifier\ asks two disjoint edges then \verifier\ learns 
nothing about the values committed by the \emph{forever-hiding}
property of the commitment scheme. The resulting $2$--prover one-round interactive 
 proof system is denoted $\twopcl$. 
 
\begin{quote}\label{MIP_3COLZK}
\rule{\linewidth}{1pt}
\prot{\,\,\twopcl[G]}{Two-prover, 3-COL} 
$\prover_{1}$ and $\prover_{2}$ pre-agree on random masks $b_{{i}}\in_R \mathbb{F}_3$
for each  ${i} \in V$ 
and a random 3-colouring of $G$: $\left\{ (i,c_{i}) | c_i \in \mathbb{F}_3\right\}_{i\in V}$ 
such that $({i},{j}) \inn E \implies c_{{j}} \neq c_{{i}}$.
%The values $n_{0},n_{1},n_{2},n_{3},r_{0},r_{1},r_{2},r_{3}$ are selected under one of two constraints: either
%$$(n_{0}, n_{1}) = (n_{2}, n_{3}), (r_{0}, r_{1}) \nneq (r_{2}, r_{3}) \text{ or }$$
%$$\exists i,j \inn \{0,1\}\times \{2,3\} : n_{i}=n_{j},r_{i}=r_{j}.$$

{\bf Commit phase:}

\begin{itemize}

\item $\verifier$ picks $(((i,j),r, s ), ((i'\!,j'), r'\!, s'))\in_{\mathcal{D}'_G} \left(E\times (\mathbb{F}^*_3)^2\right)^2$, sends
 $((i,j),r,s)$ to $\prover_1$ and $((i'\!,j'),r'\!,s')$ to $\prover_2$.
 
\item If $(i,j) \in E$ then $\prover_1$ replies $w_i = b_i\cdot r + c_i $ and $w_{j}
= b_j\cdot s + c_j$.
\item If $(i'\!,j') \in E$ then $\prover_2$ replies $w'_{i'} = b_{i'} \cdot r' + c_{i'} $ and $w'_{j'}
= b_{j'} \cdot s' + c_{j'}$.

\end{itemize}

{\bf Check phase:}

\begin{itemize}
\item[] {\bf Edge-Verification Test:}
\item if $(i,j) = (i'\!,j')$ and $(r'\!, s') \nneq (r, s)$ then \verifier\
 accept iff $w_{i}+w'_{i}\neq w_{j}+w'_{j}$.
\item[] {\bf Well-Definition Test:}
\item If $(i,j)=(i'\!,j')$ and $\neg\left((r'\!, s') \nneq (r, s)\right)$ then \verifier\ accepts iff 
$((w_i = w'_i)\vee (r\neq r'))\wedge ((w_j = w'_j)\vee (s\neq s'))$.
\item if $(i,j)\cap (i'\!,j')=i$ and $r'=r$ then \verifier\ accepts iff $w_{i} = w'_{i}$.
\item If $(i,j)\cap (i'\!,j')=j$ and $s'=s$ then \verifier\ accepts iff $w_j = w'_j$.
\end{itemize}

\rule{\linewidth}{1pt}
\end{quote}

%{Proof of soundness of the Two-Prover Protocol (over $\mathbb{F}_{3}$).}
Clearly, \twopcl\ satisfies perfect completeness.
The following theorem establishes that in addition
to perfect completeness, 
$\twopcl$ is sound against classical provers.

\begin{theorem}\label{soundnesscl}
The two-prover interactive proof system
\twopcl\   is perfectly complete with classical
value
$\omega(\twopcl[G]) \leq 1-\frac{1}{12\cdot |E|}$ 
upon any graph $G=(V,E)\notin\mathsf{3COL}$.
\end{theorem}
\begin{proof}
Perfect completeness is obvious. Assume $G\notin \mathsf{3COL}$ 
and let us consider the probability
$\delta$ that \verifier\ detects an error in the check phase
when interacting with two local dishonest provers
$\tilde{\prover}_1$ and $\tilde{\prover}_2$. 
\twopcl\ is a one-round protocol
where the provers cannot communicate directly with each other nor through
\verifier's questions since they are independent of the provers'
answers. It follows that the strategy of $\tilde{\prover}_1$ and $\tilde{\prover}_2$
can be made deterministic without damaging the soundness
error by letting each prover choosing the answer that maximizes
her/his probability of success given her/his question. Therefore,
consider a deterministic strategy as a pair of arrays $W^{\ell}[{i},r,{j},s]\in\mathbb{F}^2_3$ to be 
used by prover $\tilde{\prover}_{\ell}$ 
for $\ell\inn \{1,2\}$ (i.e. we only care about the entries where $({i},{j}) \inn E$
upon question $((i,j),r,s)$). 
For $z\in \{1,2\}$, $W_{z}^{\ell}[\cdot,\cdot,\cdot,\cdot]$ 
 is the $z$-th component 
 of the output pair $W^{\ell}[\cdot,\cdot,\cdot,\cdot]$.
 We let
$W^{\ell}_1(i,r,j,s)=W^{\ell}_2(j,s,i,r)$, as the order in which the vertices
of an edge are given to a prover is irrelevant ($V$ can always choose the same order).
We say that $W[i,r]$ for $[i,r] \in E\times \mathbb{F}^*_3$
is \emph{well defined} if  
for all $j,k$ such that $({i},{j}), ({i},{k}) \inn \text{\sf{Edges}}(i)\neq \emptyset$ 
and $\forall s,t \in \mathbb{F}^*_3$,
\begin{equation}\label{longegal}
W_{1}^{1}[{i},r,{j},s]
=  W_{1}^{2}[{i},r,{k},t]\enspace.
\end{equation}
For $W[i,r]$ well defined, we set $W[i,r]:=W_{1}^{1}[{i},r,{j},1]$
for an arbitrary $j$ such that $(i,j)\in \text{\sf{Edges}}(i)$.

We now lower bound the probability $\delta_{\text{wdt}}>0$
that, when $W[{i},r]$ is not well-defined for some ${i}\in V$ and $r\in \mathbb{F}^*_3$, 
the well-definition test will detect it. 
When (\ref{longegal}) is not satisfied , we have
$W_{1}^{1}[{i},r,{j},s] \neq W_{1}^{2}[{i},r,k,t]$
for some $(i,j), (i,k)\in \text{\sf{Edges}}(i)$.
% or
%$W_{1}^{1}[{i},r,{j},s] \neq W_{1}^{2}[{i},r,k,t]$
%or $W_{1}^{1}[{i},r,{k},t] \neq W_{1}^{2}[{i},r,{j},s]$ or
%$W_{1}^{1}[{i},r,{k},t] \neq W_{1}^{2}[{i},r,k,t]$.
Let $e=(i,j)$ and $e'=(i,k)$ be these two edges. 
%\LC{Une des dernieres corrections pour moi a ete d'enlever
%la demande ici que $e\neq e'$. Il faudrait v\'erifier que la suite
%de la preuve fonctionne lorsque $e=e'$...}
%\CCom{phrase \`a moiti\'e \'ecrite ici...}
%In each of these cases, the inequality involves
%two edges $e,e'$ where either $e=e'$ 
%or $e\cap e'=i$ 
%(i.e. there are at most $|V|$ 
%edges $e'\in E$ such that $e\cap e'=i$) 
%together with three values in $\mathbb{F}^*_3$
%to indicate the randomness for both tables (where  randomness
%$r$ is a common entry to both tables).  
According to (\ref{bbpee}) (and (\ref{bpee}) when $e=e'$),
the well-definition test will then detect an error 
with probability 
\begin{equation}\label{ppp}
\Pr{\left(\text{\verifier\ picks $e$ and $e'$ with randmoness $r,s,t$} \right)}
= p'(e,r,s,e'\!,r,t) \geq  \frac{1-\epsilon}{16\cdot |E||\text{\sf{Edges}}(i)|}   \enspace.
\end{equation}
We can do much better.
Consider 
$W_1^1[i,r,m,u], W^2_1[i,r,m,u]$ for $(i,m) \in {\sf Edges}(i)$ and $u\in\mathbb{F}^*_3$.
For $i\in V$ and and value $r\in\mathbb{F}^*_3$ fixed, three cases can happen:
\begin{enumerate}
\item $W^{1}_1[i,r,m,u] \neq W_1^2[i,r,k,t]$, in which case $e=(i,m)$ and $e'=(i,k)$ 
are incompatibe for values $u$ and $t$, or
\item $ W_1^1[i,r,j,s]\neq W^2_1[i,r,m,u]$, in which case $e=(i,j)$ and $e'=(i,m)$
are incompatible for values $s$ and $u$, or
\item $W^{1}_1[i,r,m,u] = W_1^2[i,r,k,t]$ and  $W^2_1[i,r,m,u] = W_1^1[i,r,j,s]$, in which 
case $W^{1}_1[i,r,m,u] \neq W^2_1[i,r,m,u]$ and $e=e'=(i,m)$ are incompatible 
for value $u$ on both sides.
\end{enumerate}
In other words, if $(i,j),(i,k)\in \text{\sf{Edges}}(i)$ are such that
$W_{1}^{1}[{i},r,{j},s] \neq W_{1}^{2}[{i},r,k,t]$ then for any $(i,m)\in  \text{\sf{Edges}}(i)$
and for any randomness $u\in\mathbb{F}^*_3$ associated to node $m$, 
\verifier\ catches the provers with probability expressed on the right hand side
of (\ref{ppp}).
It follows that if $W[i,r]$ is not well defined then there are $2\cdot|\mathsf{Edges}(i)|$
ways for \verifier\ to 
catch the provers  and each of these has probability at least
$\frac{1-\epsilon}{16\cdot|E|\cdot|\mathsf{Edges}(i)|}$ to be picked.
It follows that,
\begin{align*}
\delta_{\text{wdt}} &\geq  %\min_{e=(i,j)\in  E}
{\frac{2 (1-\epsilon)\cdot |\mathsf{Edges}(i)|}{16\cdot |E|\cdot 
|\mathsf{Edges}(i)|} = \frac{1-\epsilon}{8\cdot |E|}\enspace.}
\end{align*}

%\LC{Il me faut ajouter la discussion qui dit qu'on doit multiplier 
%la probabiite $p(e,e')$ par $\#\mathbb{F}^*_3\cdot \#\left( \mathsf{Edges}(i)\cup \mathsf{Edges}(j)\right)-1$.
%Pas eu le temps de terminer...}
%
%It follows from (\ref{bpee}) that if $W[i,r]$ is not well-defined 
%for some $i\in V$ and $r\in \mathbb{F}_3^*$ then 
%the well-definition test will detect it with probability 
%\[\delta_{\text{wdt}} \geq \frac{1}{8}\min_{\stackrel{e,e'\in E:}{e\cap e'\neq 0}}(p(e,e')) 
%\geq \frac{2\cdot \#\left( \mathsf{Edges}(i)\cup \mathsf{Edges}(j)\right)-1}{16\cdot |E|\cdot \#\left( \mathsf{Edges}(i)\cup \mathsf{Edges}(j)\right)} \geq \frac{1}{8\cdot |E|}\enspace.
%\]
%
%\CCom{Consid\'erons le cas $W_{1}^{1}[{i},r,{j},s] \neq W_{1}^{2}[{i},r,{k},t]$. Il va de soit que $W_{1}^{1}[{i},r,{j},s] \neq W_{1}^{2}[{i},r,{m},u]$ ou $W_{1}^{1}[{i},r,{m},u] \neq W_{1}^{2}[{i},r,{k},t]$ aussi
%pour chaque $m,u$. Il y a donc au moins $|\mathbb{F}^*| |V|$ paires $( W_{1}^{1}[{i},r,{a},b] , W_{1}^{2}[{i},r,{c},d] )$ en d\'esaccord. Tous les autres cas sont pareil.}
%

Now, assume that for all ${i}\in E$ and $r\in \mathbb{F}_3^*$, 
$W[{i},r]$ is well-defined, which means that 
the commitment values produced by the provers 
satisfy the consistency test. 
As discussed in section \ref{commitment}, 
when the commitments are consistent,    
the unique values committed upon are defined by
$c_{{i}} := 2^{-1}\cdot\left(W[{i},r]+W[{i},-r]\right)$.
Since $G\notin \mathsf{3COL}$,
 two of the nodes must be of the same colour at the end-points 
 of at least one edge $(i^*,j^*)\in E$.
  In this case the edge-verification test will detect it when
 $(i^*,j^*)$ is the edge announced to both provers and if randomness
 $(r,s)\in \mathbb{F}_3^*\times \mathbb{F}_3^*$ 
 is announced to $\prover_1$ then $(-r,-s)$ is the randomness
 announced to $\tilde{\prover}_2$.  Using (\ref{bbpee2}), the probability  $\delta_{\text{evt}}$ to 
 detect such an edge when $W[i,r]$ is well 
 defined for all $i\in V$ and $r\in\mathbb{F}_3^*$ satisfies
 \[   \delta_{\text{evt}} \geq \min_{e\in E}( p'(e,r,s,e,r,s)) \geq \frac{\epsilon}{4\cdot |E|} \enspace.
 \]
 Therefore, the detection probability $\delta$ of any deterministic strategy 
 for $G\notin \mathsf{3COL}$ satisfies
\[~~~~~~~~~~~~~~~~~~~~~~~~~~~~\delta \geq \min(\delta_{\text{wdt}}, \delta_{\text{evt}})
 \geq \frac{1}{12\cdot |E|}  ~~\text{(maximized at $\epsilon=1/3$)}\enspace.\text{~~~~~~~~~~~~~~~~~~~~~~~~~~}
\]
The result follows as the classical value of the game $\omega(\twopcl[G])\leq 1-\delta$.\qed
\end{proof}
%The following corollary is a straightforward consequence 
%of the fact that $\twopcl[G]$ for $G\notin \mathsf{3COL}$
%is a game with classical value  $\omega(\twopcl[G])$.
%\begin{corollary}\label{corr}
%For $G=(V,E) \notin \mathsf{3COL}$, $\omega(\twopcl[G]) \leq 1-\frac{1}{12\cdot|E|}$.
%\end{corollary} 

To prove (perfect) zero-knowledge, it suffices to show that if 
$((i,j),r,s)$ and $((i'\!,j'),r'\!,s')$ are selected arbitrarily, 
\verifier\
 can determine at most the colours of two nodes (that form an edge).
 The commitments prevent a dishonest prover $\widetilde{\verifier}$ 
 to learn the colours of two nodes that are not connected by an edge in $G$.
Proving this is not very hard and will be done in Section \ref{PofZK} 
for the three-prover case (although with three provers, $\widetilde{\verifier}$ may also learn
the colour of three nodes that form a triangle). The addition of a third
prover will allow, using lemma~\ref{values}, 
 to get soundness against entangled provers
without compromising zero-knowledge.
As shown in \cite{CHTW04}, their protocol is not necessarily sound against two entangled provers.
We also do not know whether  \twopstd\ is sound against two entangled provers.
%\begin{figure}[h]
%\centering
%\includegraphics[angle=0, width=0.8\textwidth]{CHTW04.png}
%\caption{Dixit CHTW-04}
%\label{CHTW04}
%\end{figure}
%\CCom{Selon le texte ci-bas, on ne peut pas affirmer que le protocole de [CHTW04] n'est pas consistant contre des prouveurs intriqu\'es, mais seulement que sa g\'en\'eralisation avec plus de couleurs ne l'est pas...}
%
%\CCom{Il faut voir s\'eparemment pour \twopstd\ et une version am\'elior\'e o\`u \verifier\ regarde toujours que les couleurs sont distinctes aux deux bouts des ar\^etes.
%Je ne crois pas arriver \`a une conclusion avant de soumettre ce papier.}

\section{Three-Prover Protocol Sound Against Entangled Provers}
\label{MIP_3COL3s}

The three-prover protocol 
\threepq, defined below,
is identical to \twopcl\  except that $\prover_3$ 
is asked to repeat exactly what $\prover_{1}$ or $\prover_{2}$
has replied. The prover that $\prover_3$ is asked to emulate
is picked at random by \verifier. An application of lemma~\ref{values}
allows to conclude the soundness of \threepq\ against 
entangled provers. Zero-knowledge remains since the only
way to provide  \verifier\ with the colours of more than two connected nodes 
is if they form a complete triangle of $G$. This reveals
nothing beyond the fact that $G\in \mathsf{3COL}$
to \verifier, since all nodes will then show different 
colours.
 
\subsection{Distribution of questions}
The probability distribution
$\mathcal{D}''_G$ for \verifier's questions to the three provers
is easily obtained from the distribution $\mathcal{D}'_G$ for the questions
in protocol $\twopcl[G]$. 
\verifier\ picks $((e,r,s),(e'\!,r'\!,s'))\in_{\mathcal{D}'_G} 
\left(E\times (\mathbb{F}^*_3)^2\right)^2$
and sets $e''=e$, $r''=r$, and $s''=s$ with probability $\frac{1}{2}$
or sets $e''=e'$, $r''=r'$, and $s''=s'$ also with probability $\frac{1}{2}$. 
Defined that way, 
 $\mathcal{D}''_G$ is a properly defined probability distribution
 for \verifier's three questions, each one in  $E\times (\mathbb{F}^*_3)^2$.
 
%A result of \cite{doi:10.1137/090751293}, allows to conclude from
%the perfect completeness of protocol~\ref{MIP_3COLZK} against classical 
%non communicating provers  that the three-prover 
%protocol is perfectly complete and is sound against entangled provers.
%
\pagebreak
\subsection{The Protocol}
\begin{quote}\label{MIP_3COL3} 
\rule{\linewidth}{1pt}
\prot{\,\,\threepq[G]}{Three-prover, 3-COL.}
%The verifiers $V_{1},V_{2},V_{3}$ pre-agree on random edges $(n_{0},n_{1})$, $(n_{2},n_{3})$ and $(n_{4},n_{5})$, random strings $r_{0},r_{1},r_{2},r_{3},r_{4},r_{5}\neq 0$ and 
%the 
Provers $\prover_{1},\prover_{2}$, and $\prover_{3}$ 
pre-agree on random values $b_{{i}}\in_R \mathbb{F}_3$ for all ${i} \in V$ 
and a random 3-colouring of $G$: $\left\{ (i,c_{i}) | c_i \in \{0,1,2\} \right\}_{i\in V}$ 
such that $({i},{j}) \inn E \implies c_{{j}} \neq c_{{i}}$.
%The values $n_{0},n_{1},n_{2},n_{3},n_{4},n_{5},r_{0},r_{1},r_{2},r_{3},r_{4},r_{5}$ are selected under one of four constraints: either 
%$$(n_{0},n_{1}) = (n_{2},n_{3}) = (n_{4}, n_{5}), (r_{0}, r_{1}) = (r_{4}, r_{5}) \nneq (r_{2}, r_{3}) \text{ or }$$
%$$(n_{0},n_{1}) = (n_{2},n_{3}) = (n_{4}, n_{5}), (r_{0}, r_{1}) \nneq (r_{4}, r_{5}) = (r_{2}, r_{3}) \text{ or }$$
%$$\exists i,j \inn \{0,1\}\times \{2,3\} : n_{i}=n_{j},r_{i}=r_{j} \text{ and } (n_{0}, n_{1}) = (n_{4}, n_{5}), (r_{0}, r_{1}) = (r_{4}, r_{5}) \text{ or }$$
%$$\exists i,j \inn \{0,1\}\times \{2,3\} : n_{i}=n_{j},r_{i}=r_{j} \text{ and } (n_{2}, n_{3}) = (n_{4}, n_{5}), (r_{2}, r_{3}) = (r_{4}, r_{5}).$$
%

{\bf Commit phase:}

\begin{itemize}
\item $\verifier$ picks $(((i,j),r, s ), ((i'\!,j'), r'\!, s'), ((i''\!, j''),r''\!,s''))\in_{\mathcal{D}''_G} \left(E\times (\mathbb{F}^*_3)^2\right)^3$, sends
 $((i,j),r,s)$ to $\prover_1$, sends $((i'\!,j'),r'\!,s')$ to $\prover_2$, and sends 
 $((i''\!,j''), r''\!,s'')$ to $\prover_3$.

\item If $(i,j)\in E$  then $\prover_1$ replies $w_i = b_i\cdot r+c_i$ and
$w_j = b_j\cdot s + c_j$.

\item If $(i'\!,j')\in E$ then $\prover_2$ replies $w'_{i'} = b_{i'} \cdot r'+c_{i'}$ and
$w'_{j'} = b_{j'} \cdot s' + c_{j'}$.

\item If $(i''\!,j'')\in E$ then $\prover_3$ replies $w''_{i''} = b_{i''} \cdot r''+c_{i''}$ and
$w''_{j''} = b_{j''} \cdot s'' + c_{j''}$.

\end{itemize}

{\bf Check phase:}

\begin{itemize}

\item[] {\bf Consistency Test:}
\item  If $((i''\!,j''),r''\!,s'') = ((i,j),r,s)$ then \verifier\ rejects if $(w_{i},w_{j}) \neq (w''_{i''},w''_{j''})$.
\item  If $((i''\!,j''),r''\!,s'') = ((i'\!,j'),r'\!,s')$ then \verifier\ rejects if $(w'_{i'},w'_{j'}) \neq (w''_{i''},w''_{j''})$.

\item[] {\bf Edge-Verification Test:}
\item if $(i,j) = (i'\!,j')$ and $(r'\!, s')\nneq (r, s)$ then \verifier\
 accept iff $w_{i}+w'_{i} \neq w_{j}+w'_{j}$.
\item[] {\bf Well-Definition Test:}
\item if $(i,j)\cap (i'\!,j')=i$ and $r=r'$ then \verifier\ accepts iff $w_{i} = w'_{i}$.
\item If $(i,j)\cap (i'\!,j')=j$ and $s=s'$ then \verifier\ accepts iff $w_j = w'_j$.
\end{itemize}

\rule{\linewidth}{1pt}
\end{quote}
In protocol \threepq, after the three questions picked according $\mathcal{D}''_G$
by \verifier\ have been answered by the the provers,
$\verifier$ accepts if and only if the replies of $\prover_1$ and $\prover_2$ 
are accepted in  $\twopcl$  and in addition, $\prover_3$ 
gave the same reply than the prover it emulates.

%\subsubsection{Proof of soundness of the Three-Prover Protocol (over $\mathbb{F}_{3}$).}
The soundness of protocol $\threepq$ against entangled provers
can easily be shown a direct consequence of the soundness
of protocol $\twopcl$ against classical provers, by an application of Lemma~\ref{values}.
Indeed, the soundness error corresponds to the quantum value
of the game when $G \notin \mathsf{3COL}$ and $\twopcl$ is obviously symmetric.   
%\begin{theorem}[\cite{doi:10.1137/090751293}, Lemma 16]
%If $\omega^{*}(G') > 1-\epsilon$, then $\omega(G) > 1-\delta$, for $\delta = \epsilon + 12 |Q| \sqrt{\epsilon}$, where $Q$ is the set of all possible questions 
%that \verifier\ may ask. \LC{Il faut ajouter pour les prots \`a une ronde sym\'etrique}
%\end{theorem}

\begin{theorem}\label{t2}
The three-prover interactive proof system
$\threepq$ is perfectly complete and has quantum value
\begin{equation}\label{delta} 
\omega^*(\threepq[G])  \leq 1-\left(  \frac{1}{25|E|}\right)^4
\end{equation}
upon any graph $G=(V,E)\notin \mathsf{3COL}$.
%If $\omega(\twopcl[G]) \leq 1-\delta$ then $\omega^{*}(\threepq[G]) 
%\leq 1-\epsilon$, for
%$\epsilon = 36 |Q|^{2} \left( -1 \pm \sqrt{ 1- \frac{\delta}{36 |Q|^{2} } } \right)^{2}$.
\end{theorem}
\begin{proof}
Assume $G=(V,E)\notin \mathsf{3COL}$.
The contrapositive of Lemma~\ref{values} indicates 
any one-round symmetric  game $\twopcl[G]$ with classical value 
$\omega(\twopcl[G])\leq 1-\delta-12|Q|\sqrt{\delta}$ is such that
the modified game
$\threepq[G]$  has quantum value
 $\omega^*(\threepq[G])\leq 1-\delta$. The set $Q$ of questions to each player 
 satisfies $|Q|=4|E|$.
Theorem~\ref{soundnesscl} establishes that $\delta + 12|Q|\sqrt{\delta}\geq \frac{1}{12 |E|}$, 
 which implies $\sqrt{\delta} \geq \frac{1}{(1+12|Q|)\cdot 12 \cdot |E|}=\frac{1}{12|E|+576 |E|^2}
 \geq \frac{1}{588|E|^2}$,
 and the result follows.
 % 
% 
% 
% $
% \frac{1}{16\cdot |E|\cdot |V|}+ 12\cdot|Q|\sqrt{\frac{1}{16\cdot |E|\cdot |V|}}$ where $Q$ is the set of all possible questions in \twopcl. 
%$$\epsilon% \approx 36 |Q|^{2} \left( 1 - 1 +  \frac{1}{2} \cdot \frac{\delta}{36 |Q|^{2}} \right)^{2}
%		\approx \left( \frac{\delta}{12 |Q|} \right)^{2}
%		\approx \left( \frac{|V|}{12 \cdot 8 |E|^{2} \cdot 4 |E|}  \right)^{2}
%		\approx \frac{|V|^{2}}{2^{14}\cdot  9  |E|^{6}}$$ 
%whenever, as above, $\delta = \frac{|V|}{8 |E|^{2}}$ or, depending on the density of the graph:
%$$\frac{1}{2^{14}\cdot  9  |V|^{10}} \leq \epsilon \leq \frac{1}{2^{14}\cdot  9  |V|^{4}}  \text{  when  } \frac{1}{8 |V|^{3}} \leq \delta \leq \frac{1}{8 |V|}.$$
\qed
\end{proof}

As an immediate consequence of Theorem~\ref{t2}, 
$\Omega(|E|^4)$ sequential repetitions of $\threepq$ produces
an  interactive proof system for $\mathsf{3COL}$  with negligible
soundness error. Although the resulting proof system can be 
implemented on short distances,
these many sequential 
communication rounds need to be performed at high rate for a given proof
to be concluded in reasonable time. A few executions of 
$\threepq$ could be run in parallel  without having to increase 
(significantly) the distances while reducing the number of sequential
rounds. However, we don't know how the soundness error decreases when
$\threepq$ is run only a few times in parallel, even though the results
of Kempe and Vidick, a quantum version of Raz's parallel repetition
theorem\cite{Raz98}, indicate that $\Omega(|E|^4)$
runs in parallel produces a proof system with negligible soundness error\cite{KV11}.

\subsection{Proof of Perfect Zero-Knowledge}\label{PofZK}
%\LC{Ici, je mettrais la preuve de zk. }
%To prove Zero-Knowledge it suffices to show that if $n_{0},n_{1},n_{2},n_{3},n_{4},n_{5},r_{0},r_{1},r_{2},r_{3},r_{4},r_{5}$ are selected arbitrarily,
%the verifiers $V_{1},V_{2},V_{3}$ can determine at most the colours of three nodes (that form a triangle in $G$) or at most the colours of two nodes (that form an edge).
%Proving this is not very hard.

%%%%%%%%%%%%%%%%%%%%%%%
%%%%%%%%%%%%%%%%%%%%%%
%\subsection{Proofs of Soundness}\label{PofS}
%
%The proof of soundness is in two parts: first prove soundness of the Two-Prover Protocol where the provers are local,
%and second use the theorem of \cite{doi:10.1137/090751293} to prove soundness of the Three-Prover Protocol and the Two-out-of-Three-Prover Protocol where the provers are entangled.
%		
%\subsection{Proofs of Zero-Knowledge}

In this section, we prove that protocol $\threepq$ is perfect zero-knowledge.
As a consequence, $\twopcl$ is also zero-knowledge since  everything
\badp{\verifier}\ sees in \twopcl\ can also be observed in $\threepq$.
The proof of zero-knowledge proceeds using the fact that a vertex must appear at least 
twice to have its colour unveiled. This is the \emph{forever hiding property} of the
commitment scheme described in Section \ref{commitment}. Notice that this 
would be enough for \badp{\verifier}\ to learn something about the colouring if
no extra condition on these three vertices is observed. In fact, 
we can easily show that only a few cases of colour disclosure are possible
and in each of these cases, 
\badp{\verifier}\ learns nothing about the colouring that it could not have 
computed on its own. \badp{\verifier}\ can only learn  colour of two connected vertices in $G$ 
and nothing else or the colours of three vertices forming a triangle in $G$.
In each of these cases, \badp{\verifier}\ learns random distinct colours 
for these vertices, which is to be expected by a valid $3$-colouring of $G$.
Let us 
show why this is enforced by the properties (see Section \ref{commitment})
of the commitment scheme. Remember that in order to 
learn the colour assigned to a vertex $i\in V$, $\badp{\verifier}$ 
must ask that vertex to at least $2$ distinct provers. Otherwise, 
$\badp{\verifier}$ sees only random values returned by the provers.
There are 7 cases of figure depending on how $\badp{\verifier}$
selects the 3 edges asked. Figure~\ref{FIG:CCMIP} shows all cases.
The 3 edges indicated for each case are the one picked by $\badp{\verifier}$.
The colours associated to white vertices remain hidden by the forever hiding
property of the commitment scheme. For these vertices, the committed
values received from the provers are just random and independent elements in $\mathbb{F}_3$.
In each of the  7 cases, the unveiled colours of the vertices are displayed
in shade of grey. We see that the only way to unveil the colour of two 
vertices (cases 2, 3, 4, 5, and 6)  is when they are connected by an edge, which means that
the colours of both vertices are random but distinct.  The only way 
for \badp{\verifier}\ to learn the colour of 3 distinct vertices is when 
they form a triangle (case 7). In this case, $\badp{\verifier}$ 
learns three random and distinct colours. Clearly, this is nothing
more than something necessarily true when $G\in \mathsf{3COL}$.
 
These properties of the commitment scheme allows,
for any quantum polynomial-time dishonest verifier $\badp{\verifier}$,
an easy simulator for $\mathbf{view}(\prover_1,\prover_2,\prover_3, \badp{\verifier},G)$
when $G\in\mathsf{3COL}$,
thus establishing that \threepq\ is perfect zero-knowledge.
 
\begin{theorem}\label{zkth}
The three-prover interactive proof system $\threepq$ is perfect
zero-knowledge against quantum verifiers.
\end{theorem}

\begin{proof}
The simulator \simulator, see Fig.~\ref{simulator}, 
 is classical given blackbox access to $\badp{\verifier}$ (and \badp{\verifier}\ can be quantum).
Consider an execution $\simulator(G)$ upon graph $G=(V,E)$.
It first picks a random permutation $\text{\sc col}[\cdot]:\mathbb{F}_3 \mapsto \mathbb{F}_3$ 
over three colours, each corresponding to a distinct element in $\mathbb{F}_3$.
Table $\text{\sc mark}[i,r]\in\{\true,\false\}$, for $i\in V$ and $r\in\mathbb{F}^*_3$, 
is initialized to $\false$ and will indicate if the output 
of a prover has already been simulated for vertex $i$ with
randomness $r$. Table $\text{\sc count}[i]$, for $i\in V$,
counts the number of times vertex $i$ has been asked so far
during the simulation. Variable $c\in \mathbb{F}_3$, initialized to $0$, indicates
the next colour index the simulator should use when a new colour
must be unveiled during the simulation. 

%\pagebreak
%All arithmetic operations are performed modulo 3.
\begin{figure}[bth]
\begin{quote}
\rule{\linewidth}{1pt}
\simu{\,\,\simulator(G)}{Simulator for \badp{\verifier}'s view 
upon graph $G$ in $\threepq$.}
%\begin{protocol}\label{SIM_3COL3A} Three-prover simulator.\end{protocol}
%{\bf Commit phase:}
\emph{All arithmetic below is performed in $\mathbb{F}_3$.}
\begin{enumerate}

\item Let $\text{\sc col}[\cdot]$ be a uniform permutation of $\mathbb{F}_3$ and let 
$c := 0$.
\item $\forall i \inn V, \forall r \inn \mathbb{F}^*_3$, 
let $\text{\sc mark}[i,r] := \false$ and  $\text{\sc count}[i] := 0$.

\item Run $\badp{\verifier}$  until it returns $((i_1,j_1),r_1, s_1)$, 
$((i_2,j_2),r_2, s_2)$, $((i_3,j_3),r_3,s_3)$.

\item {\sf For each $\ell\in \{1,2,3\}$ do:}
\begin{itemize}
\item Whenever $(i_{\ell},j_{\ell}) \inn E$ is provided by $\badp{\verifier}$,
{\sf output} $(w^\ell_{i_\ell},w^\ell_{j_\ell})\in \mathbb{F}_3\times\mathbb{F}_3$ to 
$\badp{\verifier}$, both computed as follows:
\begin{enumerate}
\item {\sf If $\neg \text{\sc mark}[i_\ell,r_{\ell}]$ then}
\begin{itemize}
\item {\sf If $\text{\sc count}[i_\ell] = 0$ then} pick $W[i_\ell,r_{\ell}]\in_R \mathbb{F}_3$.
\item {\sf If $\text{\sc count}[i_\ell] = 1$ then}
\begin{itemize}
\item $W[i_{\ell},r_{\ell}] := - \text{\sc col}[c] - W[i_\ell,-r_{\ell}]$,
\item $c := c+1$.
\end{itemize}
\item $\text{\sc count}[i_{\ell}] := 
\text{\sc count}[i_\ell]+1$.
\end{itemize}
\item {\sf If $\neg \text{\sc mark}[j_\ell,s_{\ell}]$ then}
\begin{itemize}
\item {\sf If $\text{\sc count}[j_\ell] = 0$ then} pick $W[j_\ell,s_{\ell}]\in_R \mathbb{F}_3$.
\item {\sf If $\text{\sc count}[j_\ell] = 1$ then}
\begin{itemize}
\item $W[j_{\ell},s_{\ell}] := - \text{\sc col}[c] - W[j_\ell,-s_{\ell}]$,
\item $c := c+1$.
\end{itemize}
\item $\text{\sc count}[j_{\ell}] := 
\text{\sc count}[j_\ell]+1$.
\end{itemize}
\item  $\text{\sc mark}[i_\ell,r_{\ell}] := \true$, 
$\text{\sc mark}[j_{\ell}, s_\ell] := \true$.
\item $w^\ell_{i_\ell} := W[i_\ell,r_{\ell}]$, $w^\ell_{j_\ell}:= W[j_\ell, s_\ell]$.
\end{enumerate}
\end{itemize}
\end{enumerate}

\rule{\linewidth}{1pt}
\end{quote}
\caption{Simulator for \threepq.}\label{simulator}
\end{figure}

\badp{\verifier}\ is then invoked to produce questions
$((i_\ell, j_\ell),r_\ell,s_\ell)$ for all provers $\prover_\ell, \ell\in \{1,2,3\}$. 
\simulator\ now aims at setting the values $(w^\ell_{i_\ell}, w^\ell_{j_{\ell}})$ 
for  $\prover_\ell$'s commitments. If $(i_\ell,j_\ell)\notin E$,
\simulator\ produces no value for  $(w^\ell_{i_\ell}, w^\ell_{j_{\ell}})$, exactly
as $\prover_\ell$ in \threepq. 

When  $(i_\ell,j_\ell)\in E$, \simulator\ first produces $\prover_\ell$'s
commitment $w^\ell_{i_\ell}$  for $i_\ell \in V$ and then produces $\prover_\ell$'s
commitment $w^{\ell}_{j_\ell}$ for $j_\ell\in V$. We show how to compute
$w^{\ell}_{i_\ell}$, $w^{\ell}_{j_\ell}$ is computed similarly
mutatis mutandis:
\begin{itemize}
\item if $\text{\sc mark}[i_\ell,r_\ell]$  
then \simulator\ returns the value
of  $w^\ell_{i_\ell}$ already determined for the simulation of the commitment 
of an \emph{earlier} prover $\prover_h$, $h<\ell$.
This ensures that both the commitment's \emph{consistency test} 
performed and the well-definition test
are always successful, as in $\threepq$ with honest provers.
\item if $\neg\text{\sc mark}[i_\ell,r_\ell]$  
then \simulator\ has never simulated a commitment of the colour
for vertex $i_\ell$ with randomness $r_\ell$.  The value  $\text{\sc count}[i_\ell]$
indicates the number of time prior to this value for $\ell$, vertex $i_\ell$
has been asked:
\begin{itemize}
\item If $\text{\sc count}[i_\ell]=0$ then $w^\ell_{i_\ell}\in_R \mathbb{F}_3$
is picked uniformly at random, as it should be when the commitment value
for the colour of vertex $i_\ell$ is observed in isolation. 
\item If $\text{\sc count}[i_\ell]=1$ then the colour associated to vertex
$i_\ell$ has been committed to value $w^h_{i_\ell}$ 
by an \emph{earlier} simulated prover $\prover_h$, $h<\ell$ upon randomness
$-r_{\ell}$ (otherwise, $\text{\sc mark}[i_\ell,r_\ell]=\true$). 
\simulator\ sets $w^{\ell}_{i_\ell} = -\text{\sc col}[c]-w^{h}_{i_\ell}$,
which satisfies the \emph{implicit unveiling} of random colour $\text{\sc col}[c]
= -w^{\ell}_{i_\ell} -w^{h}_{i_\ell}$. The current colour $c$ is incremented.
\end{itemize}
The value of $\text{\sc count}[i_\ell]$ is increased by one and 
$\text{\sc mark}[i_\ell, r_\ell] = \true$, as the colour of vertex $i_\ell$ with randomness
$r_\ell$ has been committed upon by the simulated prover $\prover_\ell$.
\end{itemize}  

Let $(w^1_{i_1}, w^1_{j_1})$, $(w^2_{i_2}, w^2_{j_2})$,
and $(w^3_{i_3}, w^3_{j_3})$ be all commitment values 
simulated by \simulator. As discussed above 
and shown in Fig.~\ref{FIG:CCMIP}, 
the colours of no more than 3 vertices are unveiled in the process.
\simulator\ always unveils as many different colours there
are colours unveiled to $\badp{\verifier}$. 
If  \simulator's simulated committed values unveils only
the colour of one vertex then that colour is random, as it should in this case
 in \threepq. If  
\simulator's committed values unveils the colours of exactly 2 vertices then
these 2 vertices form an edge in $G$ and the colours are two different
random colours, as it should be in \threepq. Finally, when \simulator's
committed values unveil the colours of exactly 3 vertices then 
these vertices form a triangle in $G$. The 3 colours unveiled by 
\simulator\ to $\badp{\verifier}$ are different and assigned randomly
to each of the 3 vertices, as it is in  \threepq.
Otherwise, if $w^\ell_{i}$ for $i\in V$ has been generated with only
one random value then $w^\ell_{i}$ is random and uniform in $\mathbb{F}_ 3$,
exactly as it is in \threepq\ in the same situation. 
It is now clear that,
\[ \mathbf{view}(\prover_1,\prover_2,\prover_3,\badp{\verifier}, G) = \simulator(G)\enspace,
\]
and $\threepq$ is perfect zero-knowledge.
\qed
\end{proof}

\begin{figure}[th]
\centering
\includegraphics[angle=0, height=7.3cm, width=0.8\textwidth]{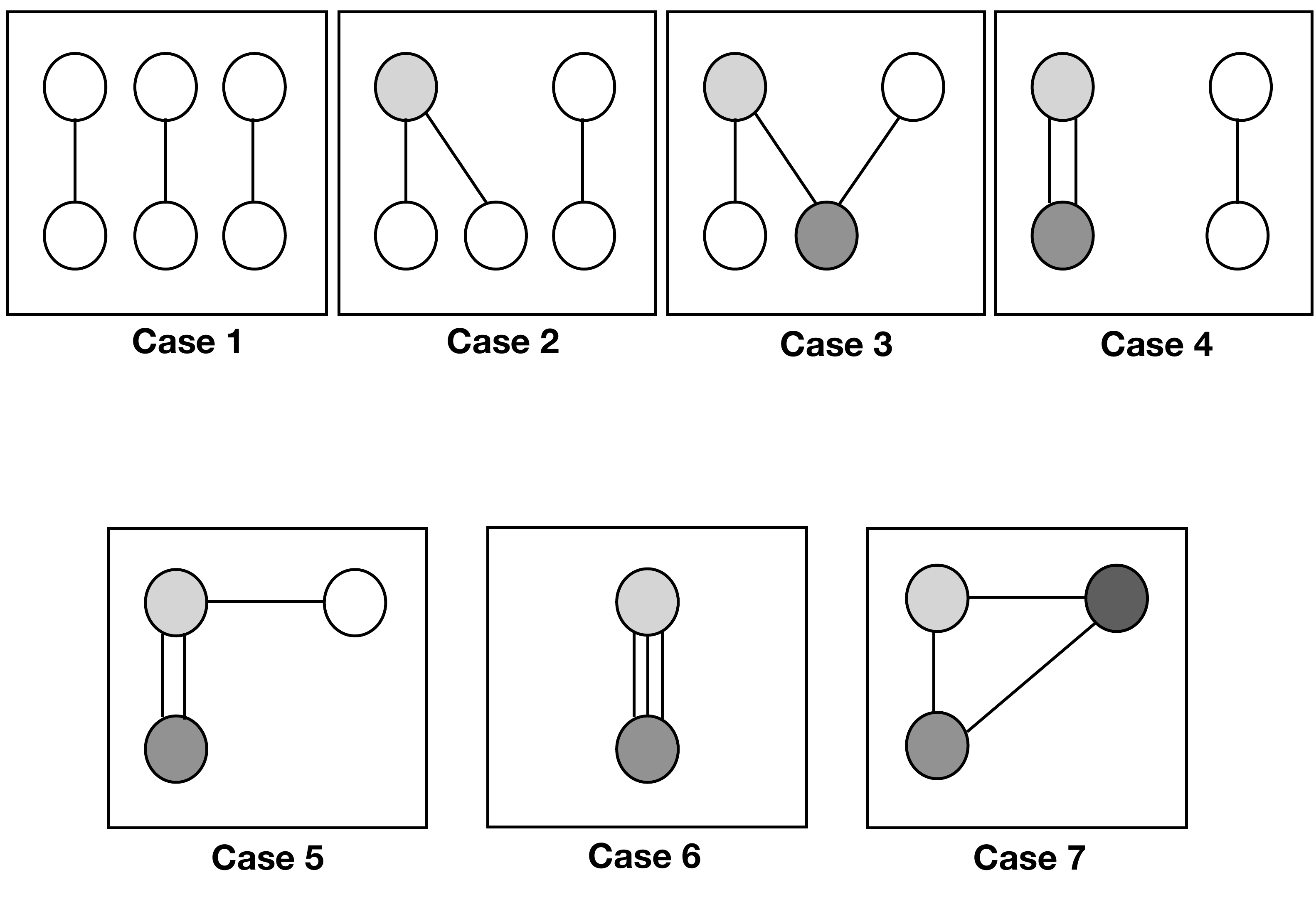}
\caption{The 7 ways to unveil the colours of at most $3$ nodes in \threepq.}\label{FIG:CCMIP}
\end{figure}

\section{Conclusion and Open Problems}

We have provided a three-prover perfect zero-knowledge 
proof system for {\bf NP} sound against entangled provers that is implementable in some 
well controlled environment. In order  to make it fully practical,
it would be better to find a protocol with smaller soundness error 
and also requiring only two provers. Is it possible? Moreover, we would like to extend
our techniques to prove any language in $\mathbf{QCMA}$ or $\mathbf{QMA}$, the natural quantum extensions of  $\mathbf{NP}$.
We would also want to prove whether \twopstd\ is sound against entangled provers.
Finally, we seek a variant of \twopstd\ that would be sound against No-Signalling provers
and a variant of \twopcl\ and \threepq\ that is both sound against No-Signalling provers and Zero-Knowledge.

%\input{3COL-2020_s07_conclusions}

%%%%%%%%%%%%%%%%%%%%%%%%%%%%%%%%%%%%%%%
%%%%%%%%%%%%%%%% END DOC %% %%%%%%%%%%%%%%%
%%%%%%%%%%%%%%%%%%%%%%%%%%%%%%%%%%%%%%%

\section*{Acknowledgements}
We would like to thank
P.~Alikhani,
N.~Brunner,
S.~Designolle,
A.~Chailloux,
A.~Leverrier,
W.~Shi,
T.~Vidick,
and
H.~Zbinden
for various discussions about earlier versions of this work. We would also like to thank Jeremy Clark for his insightful comments.

%BIB
%\bibliography{bib_MIP}
%\bibliographystyle{ieeetr}

%%%%%%%%%%%%%%%%%%%%%%%%%%
%%%%%%%%%      APPENDIX              %%%%%
%%%%%%%%%%%%%%%%%%%%%%%%%%

\end{document}